\documentclass[format=sigconf, 9pt]{acmart}

\usepackage{ellipsis, ragged2e} 
\usepackage[l2tabu, orthodox]{nag}

\usepackage{amsmath}
\usepackage{amsthm}
\usepackage{csquotes}
\usepackage{booktabs}
\usepackage{paralist}

\usepackage{bm}
\usepackage{bbm}
\usepackage{bbm}        

\usepackage{array}
\usepackage{tcolorbox}
\usepackage{mathrsfs}
\usepackage{tikz}
\usepackage{graphicx,comment,subcaption}
\usepackage{ifthen}
\usetikzlibrary{calc}
\usepackage{amsmath}
\usepackage{amsthm,amsfonts}
\usepackage[noend]{algpseudocode}
\usepackage{pgfplots}
\pgfplotsset{compat=1.15}
\usepackage{tikz}
\usetikzlibrary{shapes,positioning,arrows,arrows.meta,calc,automata,matrix,fit,backgrounds}
\tikzstyle{state}+=[minimum size = 6mm, inner sep=0,outer sep=1]

\colorlet{disabled}{lightgray}
\tikzstyle{state}=[draw,rectangle,inner sep=5pt,rounded corners=2pt]
\tikzstyle{action}=[font=\small,inner sep=0pt,outer sep=3pt]
\tikzstyle{actionnode}=[circle,draw=black,fill=black,minimum size=1mm,inner sep=0,outer sep=0]
\tikzstyle{actionedge}=[draw,-]
\tikzstyle{prob}=[font=\scriptsize,inner sep=0pt,outer sep=1pt]
\tikzstyle{probedge}=[draw,->]
\tikzstyle{directedge}=[draw,->]
\tikzset{chainarrow/.tip={Stealth[length=3pt]}}
\tikzset{>=chainarrow}

\usepackage{xparse} 
\usepackage{mathtools}
\usepackage{environ}
\usepackage{marvosym}

\usepackage{algorithmicx,algorithm}
\usepackage[noend]{algpseudocode}

\usepackage{soul}
\usepackage{todonotes} 

\usepackage[capitalize]{cleveref}

\newcommand{\B}{\mathcal{B}}

\newcommand{\C}{\mathcal{C}}
\newcommand{\calP}{\mathcal{P}}

\newcommand{\sign}{\text{sign}}
\newcommand{\calO}{\mathcal{O}}
\newcommand{\eps}{\varepsilon}
\newcommand{\defas}{\coloneqq}
\newcommand{\sadef}{\eqqcolon}

\newcommand{\ZZ}{\mathbb{Z}}

\newcommand{\RR}{\mathbb{R}}

\newcommand{\val}{\text{val}}

\newcommand{\Vertices}{V} 
\newcommand{\vertex}{v} 
\newcommand{\otherver}{u} 
\newcommand{\Edges}{E}

\newcommand{\reward}{r}
\newcommand{\discfac}{\lambda}
\newcommand{\discounted}{\text{Disc}}
\newcommand{\Plays}{\Omega} 
\newcommand{\play}{\omega} 

\newcommand{\strategyone}{\sigma}
\newcommand{\strategytwo}{\tau}

\newcommand{\Strategyone}{\Sigma}
\newcommand{\Strategytwo}{\Gamma}

\newcommand{\DSCV}{{\sc DiscVal}}
\newcommand{\TBGB}{{\sc DiscVal-Bin}}
\newcommand{\TBGD}{{\sc DiscVal-Dun}}
\newcommand{\TBGW}{{\sc DiscVal-Wun}}
\newcommand{\SI}{{\sc SI}}
\newcommand{\VI}{{\sc VI}}
\newcommand{\ROOTPOLY}{{\sc Root-Poly}}

\newtheorem{@theorem}{Theorem}[section]

\theoremstyle{remark} 

\newtheorem{remark}[@theorem]{Remark}

\Crefname{equation}{Eqn.}{Eqns.}

\author{Ali Asadi}
\affiliation{%
   \institution{ISTA}
   \city{Klosterneuburg}
   \country{Austria}}
\email{ali.asadi@ist.ac.at}

\author{Krishnendu Chatterjee}
\affiliation{%
   \institution{ISTA}
   \city{Klosterneuburg}
   \country{Austria}}
\email{krishnendu.chatterjee@ist.ac.at}

\author{Raimundo Saona}
\affiliation{%
   \institution{ISTA}
   \city{Klosterneuburg}
   \country{Austria}}
\email{raimundojulian.saonaurmeneta@ist.ac.at}

\author{Jakub Svoboda}
\affiliation{%
   \institution{ISTA}
   \city{Klosterneuburg}
   \country{Austria}}
\email{jakub.svoboda@ist.ac.at}

\title{Deterministic Sub-exponential Algorithm for Discounted-sum Games with Unary Weights}
\acmConference[LICS'24]{Logic in Computer Science}{July 8--12 July, 2024}{Tallinn, Estonia}

\begin{document}

\begin{abstract}
  Turn-based discounted-sum games are two-player zero-sum games played on finite directed graphs.
  The vertices of the graph are partitioned between player~1 and player~2.
  Plays are infinite walks on the graph where the next vertex is decided by a player that owns the current vertex.
  Each edge is assigned an integer weight and the payoff of a play is the discounted-sum of the weights of the play.
  The goal of player~1 is to maximize the discounted-sum payoff against the adversarial player~2.
  These games lie in NP $\cap$ coNP and are among the rare combinatorial problems that
  belong to this complexity class and the existence of a polynomial-time algorithm is a major open question.
  Since breaking the general exponential barrier has been a challenging problem,
  faster parameterized algorithms have been considered.
  If the discount factor is expressed in unary, then discounted-sum games can be solved
  in polynomial time.
  However, if the discount factor is arbitrary (or expressed in binary), but the weights
  are in unary, none of the existing approaches yield a sub-exponential bound.
  Our main result is a new analysis technique for a classical algorithm (namely,
  the strategy iteration algorithm) that present a new runtime bound which is
  $n^{\calO \left ( W^{1/4} \sqrt{n}  \right )}$, for game graphs with $n$ vertices and absolute weights of at most $W$.
  In particular, our result yields a deterministic sub-exponential bound for games with 
  weights that are constant or represented in unary.
\end{abstract}

\maketitle
\section{Introduction}

\smallskip\noindent{\em Turn-based discounted-sum games.}
Turn-based graph games~\cite{GH82,chandra1981alternation} are two-player infinite-duration zero-sum games played on a finite directed graph.
The vertex set is partitioned into player-1 and player-2 vertices.
At player-1 and player-2 vertices, the respective player chooses a successor vertex.
Given an initial vertex, the repeated interaction between the players generates an 
infinite walk (called \emph{play}) in the graph.
Strategies (or policies) for players provide the successor vertex choice at every player vertex.
A payoff function assigns a real value to every play. We consider a classical and well-studied
function: the discounted-sum payoff function~\cite{shapley1953StochasticGames,FV97,Puterman}.
Every edge of the graph is assigned an integer weight, and the payoff of a play
is the discounted-sum of the weights of the play.

\smallskip\noindent{\em Value problem: complexity and algorithm.}
The value at a vertex of a discounted-sum game is the maximal payoff that player~1 can ensure 
irrespective of the strategy choice of player~2.
The decision problem associated with the value (i.e., whether the value of a vertex is at 
least a given threshold) lies in NP $\cap$ coNP
(even UP $\cap$ coUP)~\cite{zwick1996complexity,condon1992ComplexityStochasticGames,CF11}.
This is among the rare combinatorial problems that lie in NP $\cap$ coNP and the existence of 
a polynomial-time algorithm is a major and long-standing open problem.
The classical algorithmic approaches to compute the values of a discounted-sum game are:
(a)~strategy iteration (or policy improvement)~\cite{FV97}; and
(b)~value iteration~\cite{zwick1996complexity,bansal12651satisficing}.
The best-known worst-case bounds for these algorithms are exponential time.
There is a randomized sub-exponential time algorithm to compute values~\cite{Lud95} with
the expected running time of $2^{\widetilde{\calO}(\sqrt{n})}$, where $n$ is the number of vertices,
and $\widetilde{\calO}(\cdot)$ hides polylogarithmic factors.

\smallskip\noindent{\em Related game problems on directed graphs.} 
There are interesting related problems for games on directed graphs,
namely, {\em parity games}~\cite{EJ91} and {\em mean-payoff games}~\cite{EM79,GKK88,zwick1996complexity}.
Parity games and mean-payoff games also lie in NP $\cap$ coNP~\cite{EJ91} (even UP $\cap$ coUP~\cite{Jur98}),
and the existence of polynomial-time algorithms are major open problems.
Parity games admit linear-time reduction to mean-payoff games~\cite{Jur98} and 
mean-payoff games admit linear-time reduction to discounted-sum games~\cite{zwick1996complexity}.
However, reductions in the converse direction are not known.
For parity games, several important algorithmic improvements have been achieved.
In particular, deterministic sub-exponential time algorithm~\cite{JPZ08} and deterministic quasi-polynomial 
time algorithm~\cite{CJKLS22}, that break the existing exponential-time barrier.
No similar algorithmic improvements have been achieved for mean-payoff games and discounted-sum games.
Indeed, no deterministic sub-exponential time algorithm is known for mean-payoff or discounted-sum games.

\smallskip\noindent{\em Parameterized algorithms and open problem.} 
Given that the algorithmic improvements for discounted-sum games have been rare, it is natural
to consider faster parameterized algorithms. 
The natural restriction to consider is the representation of the numbers related to discounted-sum
games.
There are two sources of numbers related to discounted-sum games: (a)~the weights; and (b)~the discount factor.
The results of~\cite{hansen2013strategy} establish, that if the discount factor is constant, 
then the strategy iteration algorithm runs in polynomial time (this result holds even in games with 
stochastic transitions~\cite{hansen2013strategy}).
However, when the discount factor is arbitrary but the weight function is expressed in unary,
no better bound than the exponential bound is known, and the existence of a deterministic sub-exponential 
algorithm is an important open problem.

\smallskip\noindent{\em Motivation.} 
While the problem of discounted-sum games with unary weights is theoretically interesting, there are practical motivations as well. 
Game graphs are models of reactive systems, where vertices represent states of the system, edges represent 
transitions, and players represent agents controlling different transitions. 
In analysis of reactive systems, small/constant weights are natural, e.g., when there are good and bad events
and weights represent the relative importance of the good and bad events~\cite{chatterjee2015quantitative,chatterjee2008logical}.
Since discounted-sum objectives are studied in reactive systems analysis~\cite{de2003discounting}, 
improved algorithm for such games are of practical relevance along with their theoretical importance.

\smallskip\noindent{\em Our result.}
In this work, our main result answers the above open question.
We present an improved analysis of the strategy iteration algorithm and show that, 
given a discounted-sum game with $n$ vertices and absolute
weights at most $W$, the running time is $n^{\calO \left (W^{1/4}\sqrt{n} \right )}$.
Hence, if the weights are constant or represented in unary, then the 
algorithm is a deterministic sub-exponential time algorithm.

\smallskip\noindent{\em Technical contributions.} 
We first present new bounds on the roots of polynomials with bounded integer coefficients. 
We show a sub-exponential lower bound and two upper bounds for the roots of polynomials: (i) a non-constructive sub-exponential upper bound; and (ii) an explicit quasi-polynomial upper bound. Our key insight is to associate all strategy profiles in discounted-sum games with rational polynomial functions with bounded integer coefficients. This main technical contribution establishes a connection between the lower bound for the polynomials and the running time analysis of the strategy iteration algorithm. To the best of our knowledge, such a connection has not been established before. Given the connection and our bounds on the roots of polynomials, we establish an improved running time analysis of the strategy iteration algorithm for discounted-sum games. \Cref{section: polynomials} presents the results on bounds on the roots of polynomials and \Cref{section: si} presents the key insights of analysis of strategy iteration algorithm using results of \Cref{section: polynomials}.

\smallskip\noindent{\em Related works.}
The algorithmic study of discounted-sum, mean-payoff, and parity games
have received significant attention. Below we summarize some related works.

\begin{compactitem}

\item {\em Parity games.} 
There has been a significant progress in the study of parity games.
While the classical algorithms~\cite{EJ91,Zie98} has exponential worst-case complexity,
faster exponential-time algorithms were achieved~\cite{Jur00,Schewe07}, and 
then deterministic sub-exponential~\cite{JPZ08} and deterministic quasi-polynomial time algorithms~\cite{CJKLS22} were obtained.
However, extending the algorithmic bounds from parity games to mean-payoff or discounted-sum 
games has been a major open question.

\item {\em Mean-payoff games.}
The algorithmic aspects of mean-payoff games has also been studied in several 
works~\cite{EM79,GKK88,zwick1996complexity,BCDGR11,DKZ19}.
All these algorithms have an exponential worst-case complexity. However, the classical
value iteration algorithm is pseudo-polynomial and runs in polynomial time if 
the weights are expressed in unary.

\item {\em Discounted-sum games.}
The value iteration algorithm for discounted-sum games has been studied in~\cite{zwick1996complexity,bansal12651satisficing},
and the strategy iteration algorithm has been studied in~\cite{hansen2013strategy}.
While the worst-case bound for the algorithms is exponential, if the discount-factor is 
constant, these algorithms run in polynomial time~\cite{hansen2013strategy,bansal12651satisficing}.
The value iteration and strategy iteration algorithms have an explicit dependence
on the discount factor.
An algorithm that does not depend on the discount factor is presented in~\cite{kozachinskiy2021polyhedral},
which is inspired by the algorithm of~\cite{DKZ19} for mean-payoff games.
This algorithm has a complexity of $\calO \left ((2+\sqrt{2})^n\right )$, and is exponential 
without any dependence on the weights or discount factor.

\item {\em Stochastic games.}
Discounted-sum games admit a linear reduction to stochastic games with reachability objectives~\cite{zwick1996complexity},
which are games with stochastic transitions.
The algorithmic study of stochastic games has been considered in several works~\cite{condon1990AlgorithmsSimpleStochastic,condon1992ComplexityStochasticGames,Lud95}.
However, even stochastic games with~0 and~1 weights are as hard as stochastic games
with general weights~\cite{condon1992ComplexityStochasticGames,andersson2009complexity}, 
i.e., 
parameterization by the weights is not useful. 

\end{compactitem}
In summary, none of the existing approaches break the long-standing exponential barrier for 
discounted-sum games with weights in unary and we present the first deterministic sub-exponential
bound.

\section{Preliminaries}
\label{Section: Preliminaries}

We present standard notations and definitions related to turn-based games, similar to~\cite{Zie98,chatterjee2014efficient}.

\smallskip\noindent{\bf Turn-based games.} 
A turn-based game (TBG) is a two-player finite game 
$G = (\Vertices = \Vertices_1 \sqcup \Vertices_2, \Edges)$ consisting of a finite graph with
\begin{compactitem}
    \item the set of vertices $\Vertices$, of size $n$, partitioned into player-1 vertices $\Vertices_1$ and player-2 vertices $\Vertices_2$; and
    \item the set of edges $\Edges \subseteq \Vertices \times \Vertices$, of size $m$, such that for all $\vertex \in \Vertices$, the set $\Edges(\vertex) \defas \{ \otherver \mid (\vertex, \otherver) \in \Edges \}$ is non-empty.
\end{compactitem}

\smallskip\noindent{\bf Steps and plays.} 
Given an initial vertex $\vertex_0 \in \Vertices$, the game proceeds as follows.
In each step, the player who owns the current vertex $\vertex$ chooses the next vertex from the set $\Edges(\vertex)$.
A \emph{play} is an infinite sequence of vertices $\play = \langle \vertex_0, \vertex_1, \ldots \rangle$
such that, for every step $t \ge 0$, the vertex $\vertex_{t+1} \in \Edges(\vertex_t)$. We denote by $\Plays$ 
the set of all plays, and by $\Plays_\vertex$ the set of all plays $\play = \langle \vertex_0, \vertex_1, \ldots \rangle$
where $\vertex_0 = \vertex$.

\smallskip\noindent{\bf Discounted-payoff objectives.} 
We consider TBGs with a weight or reward function $\reward \colon \Edges \to \ZZ$ that assigns a reward value $\reward(\vertex, \otherver)$ for all edges $(\vertex, \otherver) \in \Edges$.
We denote the largest absolute reward by $W \defas \max \{ |\reward(\vertex, \otherver)|  \mid  (\vertex, \otherver) \in \Edges \}$.
For a play $\play = \langle \vertex_0, \vertex_1, \ldots \rangle$ and a discount factor $\discfac \in [0, 1)$, the discounted-payoff (or simply payoff) is denoted by $\discounted_\discfac(\play) \defas \sum_{i \ge 0} \discfac^i \reward(\vertex_i, \vertex_{i+1})$.
The objective of player~1 is to maximize the payoff, while player~2 minimizes the payoff.

\smallskip\noindent{\bf Positional strategies.}
Strategies are recipes that specify how to choose the next vertex.
A {\em positional}  strategy $\strategyone \colon \Vertices_1 \to \Vertices$ for player~1 
(resp. $\strategytwo \colon \Vertices_2 \to \Vertices$ for player $2$) is a strategy which chooses
a vertex~$\strategyone(\vertex) \in \Edges(\vertex)$ whenever the play visits vertex~$\vertex$. 
A strategy profile~$(\strategyone, \strategytwo)$ is a pair of strategies for both players.
We denote by $\Strategyone^P$ and $\Strategytwo^P$ the set of all positional strategies for player~1 and player~2, respectively.
In general, strategies can depend on past history and not only the current vertex. 
However, for discounted-sum objective, positional strategies are as powerful as general strategies~\cite{condon1992ComplexityStochasticGames}.
Hence, in the sequel, every strategy is positional.

\smallskip\noindent{\bf Lasso-shaped plays given strategies in TBGs.}
We define $G^\strategyone$ as the restricted game where player~1 follows the strategy $\strategyone$.
We define $G^\strategytwo$ and $G^{\strategyone, \strategytwo}$ similarly.
Note that once both players have fixed their strategies we obtain a graph where each vertex has exactly one outgoing edge.
Given an initial vertex $\vertex$, we obtain a play $G_\vertex^{\strategyone, \strategytwo} = \langle \vertex_0, \vertex_1, \ldots \rangle$ such that $\vertex_0 = \vertex$, and for any step $t \ge 0$, $\vertex_{t+1} = \strategyone(\vertex_t)$ if $\vertex_t \in \Vertices_1$; and $\vertex_{t+1} = \strategytwo(\vertex_t)$ otherwise.
In other words, given strategies $\strategyone$ and $\strategytwo$, the obtained play $G_\vertex^{\strategyone, \strategytwo}$ is a {\em lasso-shaped}  play that 
consists in a finite cycle-free  path $\calP \defas \langle \vertex_0, \ldots, \vertex_{p-1} \rangle$ followed by a simple cycle $\C \defas \langle \vertex_p, \ldots, \vertex_{p + c - 1} \rangle$
repeated forever. 

\smallskip\noindent{\em Notation.} For simplicity, we denote by $\discounted_\discfac(G^{\strategyone, \strategytwo})$ a vector whose $\vertex$-th coordinate is the discounted payoff for vertex $\vertex$, i.e., $\discounted_\discfac(G_\vertex^{\strategyone, \strategytwo})$.

We recall the fundamental determinacy in positional strategies for TBGs with discounted-payoff objectives. 

\begin{theorem}[\cite{condon1992ComplexityStochasticGames}]
\label{Result: determinacy of tbgs}
For all TBGs $G$, vertices $\vertex$, reward functions, and discount factors $\discfac \in [0, 1)$, we have 
\[
\max_{\strategyone \in \Strategyone^P} \min_{\strategytwo \in \Strategytwo^P} \discounted_\discfac(G_{\vertex}^{\strategyone, \strategytwo}) =
\min_{\strategytwo \in \Strategytwo^P} \max_{\strategyone \in \Strategyone^P}  \discounted_\discfac(G_{\vertex}^{\strategyone, \strategytwo}) \,.
\]
\end{theorem}

\smallskip\noindent{\bf Value and optimal strategies.}
\Cref{Result: determinacy of tbgs} implies that switching the quantification order of positional strategies does not make 
a difference and leads to the unique notion of value, defined for a vertex $\vertex$ as 
\[
    \val_\discfac(\vertex) \defas \max_{\strategyone \in \Strategyone^P} \min_{\strategytwo \in \Strategytwo^P} \discounted_\discfac(G_{\vertex}^{\strategyone, \strategytwo}) \,.
\]
A strategy $\strategyone$ for player~1 is \emph{optimal} if, for all vertices $\vertex \in \Vertices$, we have that
\[
    \min_{\strategytwo \in \Strategytwo^P} \discounted_\discfac(G_\vertex^{\strategyone, \strategytwo}) = \val_\discfac(\vertex)\,.
\]
The notion of optimal strategies for player~2 is defined analogously.
Optimal strategies are guaranteed to exist for both players (\Cref{Result: determinacy of tbgs}).
Therefore, restricting the attention to positional strategies does not change the notion of value. 

\smallskip\noindent{\bf Value problem.}
The value problem for turn-based discounted-sum games is defined as follows
\begin{tcolorbox}
    \smallskip\noindent{\bf \DSCV.} Given a game $G$, a reward function $\reward$, and discount factor $\discfac$, compute the value function $\val_\discfac$.
\end{tcolorbox}

\smallskip\noindent{\bf Three variants.} 
There are three variants of the \DSCV\ problem with respect to
the representation of $\reward$ and $\discfac$.
\begin{compactitem}
    \item \TBGB: Both the reward function $\reward$ and the discount factor $\discfac$ are given in binary.
    \item \TBGD: The reward function $\reward$ is given in binary but the discount factor $\discfac$ is given in unary.
    \item \TBGW: The reward function $\reward$ is given in unary but the discount factor  $\discfac$ is given in binary.
\end{compactitem}

\section{Overview of Results}
We discuss previous results from the literature and present our main result.

\noindent{\bf Previous results.} 
The two classical algorithms for \DSCV\ are: (a)~Value Iteration (\VI); and (b)~Strategy Iteration (\SI). 
Both algorithms are iterative algorithms and the running time is a product of two factors: (i)~the number of
iterations and (ii)~the complexity of every iteration.
For both algorithms, the running time of every iteration is polynomial: (a)~$\calO(m)$ for \VI; and (b)~$\calO(mn^2\log m)$ for \SI~\cite{andersson2006improved}.
The bounds on the number of iterations are as follows.

\begin{theorem}
\label{Result: time complexity of SI for TBG-BIN}
    The following assertions hold:
    \begin{compactitem}
        \item The \VI\ algorithm solves \DSCV\ with $\calO \left (\frac{\log(W)}{1 - \discfac} + n \right )$ iterations~\cite{zwick1996complexity,bansal12651satisficing}.
        \item The \SI\ algorithm solves \DSCV\ with $\calO \left (\frac{m}{1 - \discfac} \log \frac{n}{1 - \discfac} \right )$ iterations~\cite{hansen2013strategy}.
    \end{compactitem}
\end{theorem}

\begin{remark}[Implications]
    We discuss the implications of \Cref{Result: time complexity of SI for TBG-BIN} for the variants of \DSCV.
    \begin{compactitem}
        \item For \TBGB, the above running times for \VI\ and \SI\ are exponential. 
        \item For \TBGD, the above result shows that \VI\ and \SI\ run in polynomial-time. Moreover, even for stochastic games \TBGD\ has polynomial-time upper bound~\cite{hansen2013strategy}.
        \item For \TBGW, the current bounds for the above and other existing algorithms do not break the exponential-time barrier.
    \end{compactitem}
\end{remark}
		
\smallskip\noindent{\em Lower bounds for \SI.} 
Lower bounds for \SI\ have been an active research topic.	
Exponential lower bound for \SI\ for parity games was established in~\cite{friedman2009super}, which was extended to other settings (such as randomized pivoting algorithm)~\cite{friedmann2011subexponential}.
Moreover, some complexity hardness result has also been established for \SI\ (e.g., the decision problem of whether \SI\ modifies an edge 
is known to be PSPACE-complete)~\cite{fearnley2015complexity}.

\smallskip\noindent{\bf Our result.}
In this work, in contrast to several lower bound results for \SI\ in the literature, we present an improved running time analysis 
for \SI. 
Our main result for \SI\ for \DSCV\ is as follows.

\begin{tcolorbox}
    \begin{theorem}[Main result]
    \label{Result: time complexity of SI for TBG-WUN}
    The \SI\ algorithm solves \DSCV\ with $n^{\calO \left ( W^{1/4} \sqrt{n}  \right )}$ iterations.
    \end{theorem}
\end{tcolorbox}

\begin{remark}[Implications]
A key implication of our result is that we obtain the first deterministic sub-exponential time algorithm for \TBGW.
In fact, as long as $W=O(n^{2-\eps})$ for $\eps>0$, we obtain a deterministic sub-exponential algorithm for \DSCV.
\end{remark}

\smallskip\noindent{\em Overview.}
Our analysis focuses on the difference of values between two lasso-shaped plays and uses a class of polynomials and the properties of their roots.
In \Cref{section: polynomials}, we present the results related to upper and lower bounds on the roots of polynomials.
In \Cref{section: si}, we present the improved analysis of the SI algorithm by employing the results of \Cref{section: polynomials}.

\section{Bounds on the roots of polynomials with integer coefficients}
\label{section: polynomials}
In this section, we present some bounds on the roots of polynomials with integer coefficients.
For a polynomial of degree $N$ with integer coefficients bounded by $W$, we show that the roots of the polynomial, which are not equal to 1, are at most sub--exponentially close to 1 in terms of $N$ and $W$. This result is achieved by \Cref{Result: bound on the roots of polynomials}, which is a generalization of \cite{borwein1999littlewood}. We then present non-constructive and constructive upper bounds on how close to 1 a root of the polynomial can be.

\smallskip\noindent{\bf Some illustrations of roots.} In \Cref{Figure: All roots}, we can observe the distribution of all roots for polynomials of degree at most 5 with integer coefficients ranging from -4 to 4. Additionally, \Cref{Figure: Roots around 1} illustrates the behavior of roots around 1.

\begin{figure*}[t]
  \begin{subcaptiongroup}
    \parbox[b]{.45\textwidth}{%
    \centering
    \includegraphics[width=0.28\textwidth]{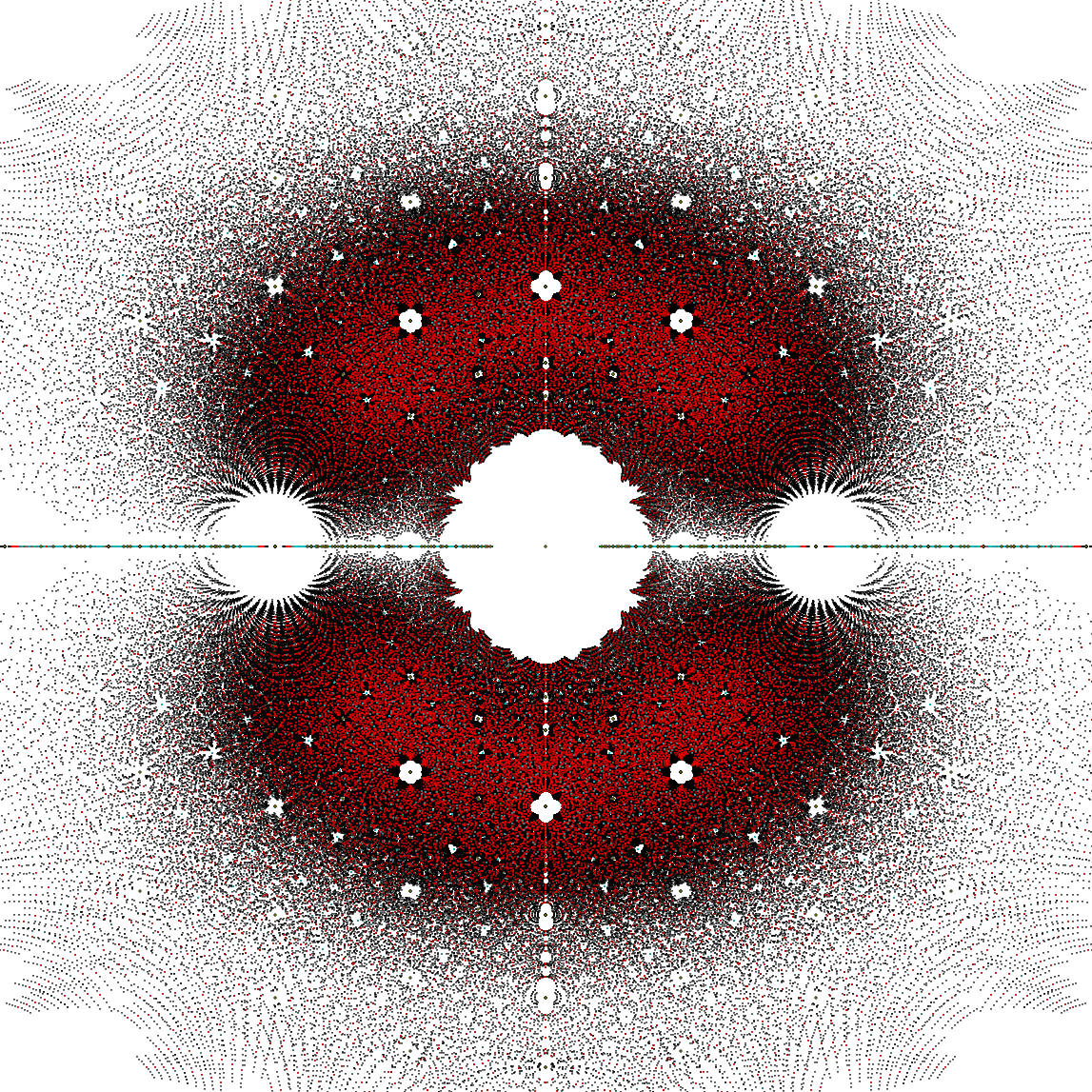}
    \caption{All roots} \label{Figure: All roots}}%
    \parbox[b]{.45\textwidth}{%
    \centering
    \includegraphics[width=0.26\textwidth]{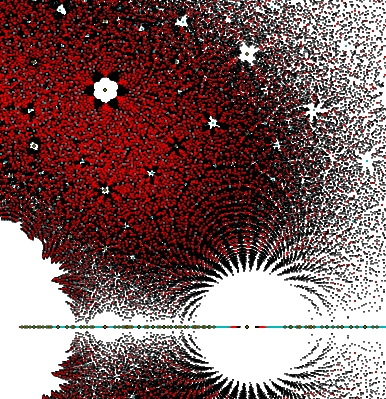}
    \caption{Roots around 1}\label{Figure: Roots around 1}}%
  \end{subcaptiongroup}
    \caption{
    Real and complex roots of all polynomials of degree at most 5 with integer coefficients ranging from -4 to 4. The horizontal axis is the real axis and the vertical axis is the imaginary axis. Roots of quadratic, cubic, quartic, and quintic polynomials are in grey, cyan, red, and black, respectively. The big hole in the middle is centered at 0, and the second biggest holes are at ±1.  (source: \cite{baez2023roots})
    }
  \label{Figure: Roots of polynomials}
\end{figure*}

\smallskip\noindent{\bf Polynomials.} A polynomial $P$ of degree $N$ with real coefficients bounded by $W$ is defined as 
\[
    P(x) = \sum_{i=0}^N a_ix^i\,,
\]
where $|a_i| \le W$ and $a_N \neq 0$. 
We denote by $\calP_N^W$ the set of all polynomials of degree $N$ with integer coefficients bounded by $W$.
We denote the degree of $P$ by $\mathfrak{d}(P)$.
We say $\alpha$ is a root of $P$ if $P(\alpha) = 0$. 
A root $\alpha$ is of order $k$ if there exists a polynomial $Q$ with rational coefficients such that $Q(\alpha) \neq 0$ and $P(x) = Q(x)(x-\alpha)^k$.
In this work, we only consider roots that are real numbers.

\smallskip\noindent{\bf Problem definition.} The problem of the roots of a polynomial is defined as follows.
\begin{tcolorbox}
    \smallskip\noindent{\bf \ROOTPOLY.} Given two positive integers $N$ and $W$, provide lower and upper bounds on 
    \[
        \inf \left \{ |1 - \alpha| \mid P \in \calP_N^W, P(\alpha) = 0, \alpha \neq 1 \right \} \,.
    \]
\end{tcolorbox}

\smallskip\noindent{\bf Previous Works.} Borwein et al. \cite{borwein1999littlewood} presents an upper and lower bound for a special variant of \ROOTPOLY\ problem. The main results are summarized as follows.
\begin{theorem}[\cite{borwein1999littlewood}]
    For a fixed non-negative integer $k$, the following assertions hold:
    \begin{compactitem}
        \item \emph{Lower bound}. Consider a polynomial $P$ of degree $N$ with $\{-1, 0, +1\}$ coefficients. If $P$ has a root of order $k$ at 1, then, for all roots $\alpha \neq 1$, we have
        \[
            |1 - \alpha| \ge \frac{4^{k+1}(k+1)!}{(N+1)^{k+2}} - \calO\left(\frac{c_1}{N^{k+3}}\right)\,,
        \]
        where $c_1 = c_1(k)$ is independent of $N$. 
        \item \emph{Upper bound}. There exists a polynomial $P$ of degree $N$ with $\{-1, 0, +1\}$ coefficients and a root $\alpha$ of $P$ such that
        \[
            |1 - \alpha| \le \frac{2^{\frac{1}{2}(k+1)(k+4)}}{N^{k+2}} + \calO\left(\frac{c_2}{N^{2k+3}}\right)\,,
        \]
        where $c_2 = c_2(k)$ is independent of $N$.
    \end{compactitem}
\end{theorem}

\begin{remark}
    Borwein et al. consider a class of polynomials with $\{-1, 0, +1\}$ coefficients and a root of order $k$ at 1. They present an asymptotic upper and lower bound on this class of polynomials when $k$ is fixed, and $N$ grows to infinity.
\end{remark}

\smallskip\noindent{\bf Our results.} 
We generalize the work of Borwein et al. \cite{borwein1999littlewood} to a class of polynomials with integer coefficients bounded by $W$, and our result is independent of $k$.

\begin{theorem}
\label{Result: bound on the roots of polynomials}
    The following assertions hold:
    \begin{enumerate}
        \item 
        \label{Result: upper bound on the roots of polynomials}
        \emph{Sub-exponential lower bound}. Consider a polynomial $P$ of degree at most $N$ with integer coefficients bounded by $W$. For all roots $\alpha \neq 1$ of $P$, we have
        \[
            |1 - \alpha| > \frac{\left \lfloor \frac{16}{7} W^{1/4} \sqrt{N} \right \rfloor!}{2W(N+1)^{\frac{16}{7} W^{1/4} \sqrt{N} + 6}}\,.
        \]
        \item 
        \label{Result: constructive lower bound on the roots of polynomials}
        \emph{Quasi-polynomial constructive upper bound}. 
        For a sufficiently large positive integer $N$, we present an explicit polynomial $P$ of degree $N$ with $\{-2, -1, 0, +1, +2\}$ coefficients and a root $\alpha < 1$ such that
        \[
            |1 - \alpha| \le 2 \left (\left \lfloor N^{2/5} \right \rfloor - 1 \right )^{-3/2 \left \lfloor \log  \left (\left \lfloor N^{2/5} \right \rfloor - 3 \right ) \right \rfloor} \,.
        \]
        \item 
        \label{Result: nonconstructive lower bound on the roots of polynomials}
        \emph{Sub-exponential nonconstructive upper bound}.
        For a sufficiently large positive integer $N$, there exists a polynomial $P$ of degree $N$ with $\{-2, -1, 0, +1, +2\}$ coefficients and a root $\alpha < 1$ such that
        \[
            |1 - \alpha| \le 2 \left ( \left \lfloor N^{2/5} \right \rfloor - 1 \right )^{-3/2\left\lfloor \sqrt{\frac{\left \lfloor N^{2/5} \right     \rfloor - 3}{\log (N^{2/5})}} \right\rfloor}  \,.
        \]
    \end{enumerate}
\end{theorem}
\Cref{section: lower bound poly} proves the lower bound, and \Cref{section: upper bound poly} proves the upper bounds given in \Cref{Result: bound on the roots of polynomials}.

\subsection{Lower bound}
\label{section: lower bound poly}
In this section, we show the lower bound on the roots of polynomials with integer coefficients. 
The proof relies on two components: (a)~a lower bound for polynomials with a root of order $k$ at~1 (\Cref{Result: bound on the roots of polynomials with order}); and (b)~an upper bound on the order of~1 as a root (\Cref{Result: bound on order of root at 1}).
These results yield \Cref{Result: bound on the roots of polynomials}-(\ref{Result: upper bound on the roots of polynomials}).

\begin{lemma}
\label{Result: bound on the roots of polynomials with order}
    Consider a polynomial $P$ of degree at most $N$ with integer coefficients bounded by $W$. If $P$ has a root of order at most $k$ at 1, then, for all roots $\alpha \neq 1$, we have
    \[
        |1 - \alpha| > \frac{(k+1)!}{2W(N+1)^{k+2}}\,.
    \]
\end{lemma}
\begin{proof}
    Consider $P = \sum_{i=0}^N a_ix^i$. We denote by $P^j$ the $j$-th derivative of $P$. Note that 
    \[
        P^j(1) = \sum_{i=j}^N j! \binom{i}{j} a_i\,.
    \]
    It follows that $P^j(1)$ is an integer divisible by $j!$, and we have
    \begin{equation}
    \label{Eq: Bound on the derivative of polynomial}
        |P^j(1)| \le W(N+1)^{j+1}\,. 
    \end{equation}
    We consider the Taylor expansion of $P$ around 1. We have 
    \[
        P(x) = P(1) + \sum_{j = 1}^N \frac{P^j(1)}{j!} (x-1)^j\,.
    \]
    We know that $P(1) = 0$, $P(\alpha) = 0$, and $P^j(1) = 0$ for all $j < k$. Without loss of generality, we assume 
    \[
        |\alpha - 1| \le \frac{1}{N+1}\,,
    \]
    otherwise the result follows immediately. By algebraic manipulation and triangle inequality, we get 
    
    \begin{align*}
        \frac{|P^k(1)|}{k!} &\le \sum_{j=k+1}^N \frac{|P^j(1)|}{j!} |\alpha - 1|^{j-k}\\
        &\le \sum_{j=k+1}^N \frac{W(N+1)^{j+1}}{j!}|\alpha-1|^{j-k} & \left ( \textsc{\Cref{Eq: Bound on the derivative of polynomial}}\right )\\
        &\le W(N+1)^{k+2}|\alpha-1|\sum_{j=k+1}^N \frac{1}{j!} & \left (|\alpha - 1| \le \frac{1}{N+1} \right ) \\
        &< \frac{2W(N+1)^{k+2}}{(k+1)!}|\alpha-1| \,. & \left (\sum_{j=k+1}^N \frac{1}{j!} < \frac{2}{(k+1)!} \right )
    \end{align*}
    We know that $P^k(1) \neq 0$, is integer, and divisible by $k!$. Therefore, 
    \[
        \frac{|P^k(1)|}{k!} \ge 1\,.
    \]
    Hence,
    \[
        \frac{(k+1)!}{2W(N+1)^{k+2}} < |\alpha - 1|\,,
    \]
    which completes the proof.
\end{proof}

To prove an upper bound on the order of~1 as a root (\Cref{Result: bound on order of root at 1}), we first present Chebyshev polynomials and their basic property, and then, we present a lemma on the existence of a specific class of polynomials (\Cref{Result: existence of a specific polynomial}), and finally, we prove \Cref{Result: bound on order of root at 1}.

\smallskip\noindent{\bf Chebyshev Polynomials.} We denote by $T_t$ the Chebyshev polynomial of degree $t$ defined recursively as follows
\begin{align*}
    &(1)\ \  T_0(x) = 1,\\
    &(2)\ \  T_1(x) = x,\\
    &(3) \ \  T_{t+1}(x) = 2xT_t(x) - T_{t-1}(x) \,.
\end{align*}
\begin{lemma}[Folklore]
\label{Result: property of Chebyshev polynomials}
    For a positive integer $t$, the Chebyshev polynomial $T_t$ satisfies that, for all $\theta \in \RR$, we have $T_t(\cos \theta) = \cos t \theta$.
\end{lemma}

\begin{proof}
    We present the proof for completeness. The proof proceeds by induction. 
    
    \smallskip\noindent{\em Base case $t = 0,1$.} 
    We have $T_0(\cos \theta) = 1$ and $T_1(\cos \theta) = \cos \theta$, which completes the case.

    \smallskip\noindent{\em Induction case $t > 1.$}
    We have
    \begin{align*}
        T_{t+1}(\cos \theta) &= 2 \cos \theta T_t(\cos \theta) - T_{t-1}(\cos \theta)\\
        &= 2 \cos \theta \cos t \theta - \cos (t-1) \theta\\
        &= \cos (t+1) \theta \,,
    \end{align*}
    where in the first equality we use the definition of $T_{t+1}$, in the second equality we use induction, and in the third equality we use $2 \cos x \cos y = \cos (x + y) + \cos (x - y)$, which concludes the induction case and yields the result.
\end{proof}

\begin{lemma}
\label{Result: existence of a specific polynomial}
    For every positive integers $N$ and $W$, there exists a polynomial $F$ of degree $k$ where $k \ge \left \lfloor \frac{16}{7} W^{1/4} \sqrt{N}  \right \rfloor + 4$ such that 
    \[
        F(0) > W\sum_{i=1}^N |F(i)|\,.
    \]
\end{lemma}
\begin{proof}
    We define
    \[
        \mu \defas \left \lfloor \frac{4}{7}W^{1/4}\sqrt{N} \right \rfloor + 1, \quad g(x) \defas \frac{1}{2}T_0(x) + \sum_{t=1}^{\mu} T_t(x)\,,
    \]
    where $T_t$ denotes the Chebyshev polynomial of degree $t$. Note that $g(1) = \mu + \frac{1}{2}$, and for $0 < x \le \pi$, we have
    \begin{align*}
        g(\cos x) &= \frac{1}{2} + \sum_{t=1}^{\mu} \cos tx\\
        &= \Re \left (\sum_{t=0}^{\mu} e^{itx} \right ) - \frac{1}{2}\\
        &= \Re \left (\frac{e^{i(\mu + 1) x} - 1}{e^{ix} - 1} \right ) - \frac{1}{2}\\
        &= \Re \left (\frac{e^{\frac{i(\mu + 1) x}{2}}}{e^{\frac{i x}{2}}}\frac{e^{\frac{i(\mu + 1) x}{2}} - e^{-\frac{i (\mu + 1) x}{2}}}{e^{\frac{i x}{2}} - e^{-\frac{i x}{2}}} \right ) - \frac{1}{2}\\
        &= \frac{2 \Re \left (e^{\frac{i \mu x}{2}} \right ) \sin \frac{(\mu + 1) x}{2} - \sin \frac{x}{2}}{2 \sin \frac{x}{2}}\\
        &= \frac{2 \cos \frac{\mu x}{2} \sin \frac{(\mu + 1) x}{2} - \sin \frac{x}{2}}{2 \sin \frac{x}{2}}\\
        &= \frac{\sin \left ( \left (\mu+\frac{1}{2} \right )x \right )}{2 \sin \frac{x}{2}}\\
        &= \frac{\sin \left ( \left (\mu + \frac{1}{2} \right )x \right )}{\sqrt{2(1 - \cos x})} \,,
    \end{align*}
    where in the first equality we use the definition of $g(\cos x)$, in the second equality we use $\cos x = \Re \left (e^{ix} \right )$, in the third equality we use geometric sum, in the fourth equality we use algebraic manipulation, in the fifth equality we use $\sin x = \frac{e^{ix} - e^{-ix}}{2i}$, in the sixth equality we use $\cos x = \Re \left (e^{ix} \right )$, in the seventh equality we use $2 \cos x \sin y = \sin (x + y) + \sin (x - y)$, and in the eighth equality we use $\sin^2 \frac{x}{2} = \frac{1 - \cos x}{2}$.
    Therefore, for all $x \in (-1, 1]$, we have 
    \[
        |g(x)| \le \frac{1}{\sqrt{2(1 - x)}}\,.
    \]
    We define 
    \[
        F(x) \defas \left ( g \left (1 - \frac{2x}{N} \right ) \right )^4, \quad k \defas 4\mu \le \left \lfloor \frac{16}{7}W^{1/4}\sqrt{N} \right \rfloor + 4\,.
    \]
    Then, $F$ is a polynomial of degree $k$. We show that 
    \[
        F(0) \ge W \sum_{i=1}^N |F(i)|\,.
    \]
    Indeed, we have
    \begin{align*}
        W \sum_{i=1}^N |F(i)| &\le W \sum_{i=1}^N \left (\frac{4i}{N} \right )^{-2} & \left (|g(x)| \le \frac{1}{\sqrt{2(1 - x)}} \right )\\
        &= \frac{WN^2}{16} \sum_{i=1}^N \frac{1}{i^2}  & (\text{rearrange}) \\
        &< \frac{\pi^2}{96} WN^2 & \left (\sum_{i=1}^\infty \frac{1}{i^2} = \frac{\pi^2}{6} \right )\\
        &\le \mu^4 & \left (\mu = \left \lfloor \frac{4}{7} W^{1/4}\sqrt{N} \right \rfloor + 1 \right )\\
        &< F(0)\,, & \left (F(0) = \left (\mu + \frac{1}{2} \right)^4 \right )
    \end{align*}
    which concludes the proof.
\end{proof}

\begin{lemma}
\label{Result: bound on order of root at 1}
    Consider a polynomial $P$ of degree at most $N$ with integer coefficients bounded by $W$. If $P$ has a root of order $k$ at 1, then 
    \[
        k \le \left \lfloor \frac{16}{7} W^{1/4} \sqrt{N} \right \rfloor + 4 \,.
    \]
\end{lemma}
\begin{proof}
    Let $P(x) = \sum_{i=0}^N a_ix^i$.
    Consider wlog that $a_0 \neq 0$. Indeed, divide $P$ by the monomial $x$ until the first coefficient of the resulting polynomial is not 0. Note that this operator  does not change the order of root at 1.
    We claim that for all polynomials $F$ of degree $(k-1)$, we have 
    \begin{equation}
    \label{equation: polynomial coefficents}
        \sum_{i=0}^N a_iF(i) = 0.
    \end{equation}
    Indeed, let $F(x) = \sum_{j=0}^{k-1} b_jx^j$. Then,
    \begin{align*}
        \sum_{i=0}^N a_iF(i) = \sum_{i=0}^N a_i \sum_{j=0}^{k-1} b_j i^j = \sum_{j=0}^{k-1} b_j \sum_{i=0}^N a_i i^j \,.
    \end{align*}
    Therefore, it is enough to show that if $j < k$, then $\sum_{i=0}^N a_i i^j = 0$.
    Consider the operator $\left (x\frac{d \cdot}{dx} \right )$. Observe that 
    \[
        x\frac{dP}{dx} = \sum_{i=0}^N ia_ix^i\,.
    \]
    Therefore, by applying this operator $j$ times on $P$, we get a polynomial $Q \defas \sum_{i=0}^N i^ja_ix^i$.
    Note that $Q$ has a root at 1, i.e., 
    \[
        Q(1) = \sum_{i=0}^N i^ja_i = 0\,.
    \]
    Hence, for all polynomials $F$ of degree $(k-1)$, we have 
    \[
        \sum_{i=0}^N a_iF(i) = 0\,.
    \]
    For the sake of contradiction, assume $P$ has a root at 1 of order at least 
    \[
        \left (\left \lfloor \frac{16}{7} W^{1/4} \sqrt{N} \right \rfloor + 5 \right )\,.
    \]
    On the one hand, \Cref{Result: existence of a specific polynomial} shows that there exists a polynomial~$F$ of degree $(k-1)$ such that 
    \[
        F(0) > \sum_{i=1}^N W|F(i)|\,.
    \]
    On the other hand, we know that $a_0 \neq 0$. By \Cref{equation: polynomial coefficents}, 
    \[
        \sum_{i=0}^N a_iF(i) = 0\,.
    \]
    Therefore, by rearranging and triangle inequality, we get
    \begin{align*}
        F(0) &\le |a_0||F(0)| \le \sum_{i=1}^N |a_iF(i)| \le W\sum_{i=1}^N |F(i)|\,, & (|a_i| \le W)
    \end{align*}
    which contradicts with the property of polynomial $F$. Therefore, the order of root at 1 of $P$ is at most 
    \[
        \left \lfloor \frac{16}{7} W^{1/4} \sqrt{N} \right \rfloor + 4\,,
    \]
    which concludes the proof.
\end{proof}

\begin{proof}[Proof of \Cref{Result: bound on the roots of polynomials}-(\ref{Result: upper bound on the roots of polynomials})]
    By \Cref{Result: bound on order of root at 1}, we know that $P$ has a root at 1 of order at most 
    \[
        \left ( \left \lfloor \frac{16}{7} W^{1/4} \sqrt{N} \right \rfloor + 4 \right )\,.
    \]
    By \Cref{Result: bound on the roots of polynomials with order}, for all roots $\alpha \neq 1$ of $P$, we have
    \begin{align*}
            |1 - \alpha| &> \frac{\left (\left \lfloor \frac{16}{7} W^{1/4} \sqrt{N} \right \rfloor + 5 \right )!}{2W(N+1)^{\left \lfloor \frac{16}{7} W^{1/4} \sqrt{N} \right \rfloor + 6}} & (\text{\Cref{Result: bound on the roots of polynomials with order}})\\
            &> \frac{\left \lfloor \frac{16}{7} W^{1/4} \sqrt{N} \right \rfloor!}{2W(N+1)^{\frac{16}{7} W^{1/4} \sqrt{N} + 6}}\,,
    \end{align*}
    which concludes the proof.
\end{proof}

\subsection{Upper bounds}\label{section: upper bound poly} 
In this section, we show two upper bounds on the roots of polynomials with integer coefficients (\Cref{Result: bound on the roots of polynomials}-(\ref{Result: constructive lower bound on the roots of polynomials},\ref{Result: nonconstructive lower bound on the roots of polynomials})). 
Given a polynomial with a root of order $k$ at 1, we construct a new polynomial with a root close to 1 depending on $k$ (\Cref{Result: lower bound on the roots of polynomials given a polynomial}).
For constructive upper bound, we present an explicit polynomial, and for nonconstructive upper bound, we use the existence of a polynomial with a root at 1 of higher order than the previous case.

\begin{lemma}
\label{Result: lower bound on the roots of polynomials given a polynomial}
    Consider a polynomial $F$ with $\{-1, 0, +1\}$ coefficients such that $F(x) = (x - 1)^k f(x)$, where $k$ is an integer greater or equal to 9 and $f$ is a polynomial with integer coefficients such that $f(1) \neq 0$.
    For every positive integer $d \ge \left (\mathfrak{d}(F) + 2 \right )^{3/2}$, there exists a polynomial $P$ of degree $\left (d \left (\mathfrak{d}(F) + 2 \right ) - 1 \right )$ with $\{-2, -1, 0, +1, +2\}$ coefficients and a root $\alpha < 1$ such that 
    \[
        |1 - \alpha| \le 2 \cdot d^{-(k+2)}\,.
    \]
\end{lemma}


\begin{proof}
    We define polynomial $\widehat{F}(x) \defas (x - 1)F(x)$. We claim that the desired polynomial is given by
    \[
         H(x) \defas \widehat{F} \left (x^d \right ) \left (\frac{1 - x^d}{1 - x} \right ) + F(x) \,.
    \]
    We define 
    \[
        N \defas d(\mathfrak{d}(F) + 1) + d - 1\,.
    \]
    Note that $H$ is a polynomial of degree $N$ with $\{-2, -1, 0, +1, +2\}$ coefficients and has a root of order $k$ at 1. 
    We show that $H$ has a root $\alpha < 1$ such that 
    \[
        |1 - \alpha| \le \frac{2}{d^{k+2}}\,.
    \]
    We take Taylor expansion of $H$ around 1 and get
    \[
        H(x) = (x - 1)^k \left [ f(1) + (x - 1)f^1(1) + (x - 1)d^{k+2} f(1) + E(x)  \right ] \,,
    \]
    where $E(x)$ is the residue polynomial. 
    It is enough to show that $H$ changes its sign in the interval 
    \[
        \left (1 - \frac{2}{d^{k+2}}, 1 \right )\,,
    \]
    which implies the existence of the desired root. 
    The three key terms are $E(x)$, $f(1)$, and $f^1(1)$. We bound these three terms to deduce the existence of the desired root.
    Note that 
    \[
        E(x) = \sum_{j=k+2}^{N} (x - 1)^{j-k} \frac{H^j(1)}{j!}\,
    \]
    Therefore, for $|x - 1| \le \frac{1}{N+1}$, we have
    \begin{align*}
        |E(x)| &\le \sum_{j=k+2}^{N} |x - 1|^{j-k} \frac{|H^j(1)|}{j!} & (\text{triangle inequality})\\
        &\le \sum_{j=k+2}^{N} |x-1|^{j-k} \frac{2(N+1)^{j+1}}{j!} &  \left (|H^j(1)| \le 2(N+1)^{j+1} \right )\\
        &\le 2|x - 1|^2 (N+1)^{k+3} \sum_{j=k+2}^{N} \frac{1}{j!} & \left (|x - 1| \le \frac{1}{N+1} \right )\\
        &\le \frac{4(N+1)^{k+3}}{(k+2)!}|x - 1|^2\,.  & \left (\sum_{j=k+2}^N \frac{1}{j!} < \frac{2}{(k+2)!} \right )
    \end{align*}
    We denote the sign function by $\text{sign}(x)$ defined by
    \[
    \sign(x) =
    \begin{cases}
    -1 & \text{if } x < 0, \\
    0 & \text{if } x = 0, \\
    1 & \text{if } x > 0.
    \end{cases}
    \]
    Without loss of generality, we consider $f(1) \ge 1$ (recall that $f$ is a polynomial with integer coefficients). We define 
    \[
        \beta \defas 1 - \frac{2}{d^{k+2}}\,.
    \]
    We claim that 
    \[
        \sign \left (H (\beta) \right ) = -\lim_{x \to 1^{-}} \sign(H(x))\,.
    \]
    It is enough to show that 
    \begin{align*}
        \lim_{x \to 1^{-}} &\sign \left (f(1) + (x - 1) f^1(1) + (x - 1)d^{k+2}f(1) + E(x) \right ) \\
        &= -\sign \left (f(1) + (\beta - 1)f^1(1) + (\beta - 1)d^{k+2}f(1) + E(\beta)  \right ) \,.
    \end{align*}
    Indeed, for the LHS, the term $f(1)$ is dominating. Hence, we get
    \begin{equation}
    \label{Eq: RHS of Sign of H}
        \lim_{x \to 1^{-}} \sign \left ( f(1) + (x - 1) f^1(1) + (x - 1)d^{k+2}f(1) + E(x) \right ) = 1\,.
    \end{equation}
    For the RHS, we have 
    \begin{align*}
        f(1) + (&\beta - 1)f^1(1) + (\beta - 1)d^{k+2}f(1) + E(\beta) \\
        &= \frac{-2f^1(1)}{d^{k+2}} - f(1) + E(\beta)\\
        &\le \frac{-2f^1(1)}{d^{k+2}} - f(1) + \frac{4(N+1)^{k+3}}{(k+2)!} (1 - \beta)^2\\
        &= \frac{-2f^1(1)}{d^{k+2}} - f(1) + \frac{16(N+1)^{k+3}}{(k+2)!\,d^{2k+4}}\\
        &\le \frac{-2f^1(1)}{d^{k+2}} - 1 + \frac{16(N+1)^{k+3}}{(k+2)!\,d^{2k+4}}\\
        &\le \frac{2(\mathfrak{d}(F))^{k+2}}{d^{k+2}(k+1)!} - 1 + \frac{16(N+1)^{k+3}}{(k+2)!\,d^{2k+4}}\\
        &\le \frac{2}{(k+1)!} - 1 + \frac{16(N+1)^{k+3}}{(k+2)!\,d^{2k+4}}\\
        &\le \frac{2}{(k+1)!} - 1 + \frac{16}{(k+2)!}\\
        &< 0\,, 
        \stepcounter{equation}\tag{\theequation}\label{Eq: LHS of Sign of H}
    \end{align*}
    where in the first equality we use the definition of $\beta$, in the first inequality we use 
    \[
        E(\beta) \le \frac{4(N+1)^{k+3}}{(k+2)!}(\beta - 1)^2\,,
    \]
    in the second equality we again use the definition of $\beta$, in the second inequality $f(1) \ge 1$, in the third inequality we use 
    \[
        |f^1(1)| = \frac{|F^{k+1}(1)|}{(k+1)!} \le \frac{(\mathfrak{d}(F)^{k+2})}{(k+1)!}\,,
    \]
    in the fourth inequality we use $\mathfrak{d}(F) \le d$, in the fifth equality we use $(N+1)^{3/5} \le d$ and $9 \le k$, and in the sixth inequality we again use $9 \le k$. Therefore, we have
    \[
        \sign\left (f(1) + (\beta - 1)f^1(1) + (\beta - 1)d^{k+2}f(1) + E(\beta)  \right ) = -1\,.
    \]
    By combining \Cref{Eq: RHS of Sign of H,Eq: LHS of Sign of H}, we get
    \[
        \lim_{x \to 1^{-}} \sign(H(x))= - \sign \left (H (\beta) \right ) \,.
    \]
    Hence, by continuity of $H$, there exists a root $\alpha$ of polynomial $H$ such that 
    \[
        |\alpha - 1| \le \frac{2}{d^{k+2}}\,,
    \]
    which concludes the proof.
\end{proof}

\begin{proof}[Proof of \Cref{Result: bound on the roots of polynomials}-(\ref{Result: constructive lower bound on the roots of polynomials})]
    Fix 
    \[ 
        M \defas \left \lfloor N^{2/5} \right \rfloor - 3,\quad d \defas \left \lceil (M+2)^{3/2} \right \rceil\,.
    \]
    We define the polynomial 
    \[
        P \defas x^{M - 2^{\lfloor \log M \rfloor} + 1}\prod_{i=0}^{\lfloor \log M \rfloor - 1} \left (x^{2^{i}} - 1 \right ) \,.
    \] 
    Note that $P$ is a polynomial of degree $M$ with $\{-1, 0, +1\}$ coefficients and has a root of order $\lfloor \log M \rfloor$ at 1. By \Cref{Result: lower bound on the roots of polynomials given a polynomial}, there exists a polynomial $Q$ of degree $d(M+2) - 1$ with $\{-2, -1, 0, +1, +2\}$ coefficients and a root $\alpha < 1$ such that
    \begin{align*}
        1 - \alpha &\le \frac{2}{d^{\left \lfloor \log M \right \rfloor + 2}} \le \frac{2}{d^{\left \lfloor \log M \right \rfloor}} & (\text{\Cref{Result: lower bound on the roots of polynomials given a polynomial}})\\
        &\le \frac{2}{(M+2)^{\frac{3}{2}\lfloor \log M \rfloor}} & \left (d = \left \lceil (M+2)^{3/2} \right \rceil \right )\\
        &\le \frac{2}{\left ( \left \lfloor N^{2/5} \right \rfloor - 1 \right )^{\frac{3}{2}\left \lfloor \log  \left (\lfloor N^{2/5} \rfloor - 3 \right ) \right \rfloor}}\,,& \left (M =  \left \lfloor N^{2/5} \right \rfloor - 3 \right )
    \end{align*}
    which yields the result.
\end{proof}

For the nonconstructive upper bound, we recall a result of \cite{borwein1999littlewood} on the existence of polynomials with high order roots at 1.
\begin{lemma}[\protect{\cite[Theorem~2.7]{borwein1999littlewood}}]
\label{Result: existence of polynomials with high order roots at 1}
    For a positive integer $N \ge 9$, there exists a polynomial $P$ of degree $N$ with $\{-1, 0, +1\}$ coefficients such that $P$ has a root at 1 of order 
    \[
        \left (\left \lfloor \sqrt{\frac{N}{\log N}} \right \rfloor - 2 \right )\,.
    \]
\end{lemma}

\begin{proof}
    We proceed with the proof by pigeonhole principle. Let $A_N$ be the set of all polynomials of degree at most $N$ with $\{0, 1\}$ coefficients. Given a positive integer $k$, for all polynomials $Q \in A_N$, we define the mapping $Q \mapsto \left (Q(1), Q^1(1), \ldots, Q^{k-1}(1) \right )$, where $Q^j(1)$ is the $j$-th derivative of $Q$ at 1. On one hand, the number of different outputs of the mapping is at most 
    \begin{align*}
        &\prod_{j = 0}^{k-1} (N+1)^{j+1} & (\textsc{\Cref{Eq: Bound on the derivative of polynomial}})\\
        &= (N+1)^{k(k+1)/2}\,.
    \end{align*}
    On the other hand, the number of polynomials in $A_N$, i.e., $|A_N|$ is $2^{N+1}$. Therefore, if 
    \[
        (N+1)^{k(k+1)/2} < 2^{N+1}\,,
    \]
    then there exist two polynomials $Q_1, Q_2 \in A_N$ such that for all $0 \le j < k$, we have $Q_1^j(1) = Q_2^j(1)$, which implies that the polynomial $Q_1 - Q_2$ has a root at 1 of order at least $k$. Observe that the coefficients of the polynomial $Q_1 - Q_2$ belong to the set $\{-1, 0, +1\}$. By setting $k = \left (\left \lfloor \sqrt{\frac{N}{\log N}} \right \rfloor - 2 \right )$, we obtain $(N+1)^{k(k+1)/2} < 2^{N+1)}$, which completes the proof.
\end{proof}

\begin{proof}[Proof of \Cref{Result: bound on the roots of polynomials}-(\ref{Result: nonconstructive lower bound on the roots of polynomials})]
    Fix 
    \[
        M \defas \left \lfloor N^{2/5} \right \rfloor - 3,\quad d \defas \left \lceil (M+2)^{3/2} \right \rceil\,
    \]
    By \Cref{Result: existence of polynomials with high order roots at 1}, there exists a polynomial $P$ of degree $M$ with $\{-1, 0, +1\}$ coefficients and a root at 1 of order 
    \[
        \left (\left \lfloor \sqrt{\frac{M}{\log M}} \right \rfloor - 2 \right )\,.
    \]
    By \Cref{Result: lower bound on the roots of polynomials given a polynomial}, there exists a polynomial $Q$ of degree $d(M+2) - 1$ with $\{-2, -1, 0, +1, +2\}$ coefficients and a root $\alpha < 1$ such that
    \begin{align*}
        1 - \alpha &\le \frac{2}{d^{\left \lfloor \sqrt{\frac{M}{\log M}} \right \rfloor}} & (\text{\Cref{Result: lower bound on the roots of polynomials given a polynomial}})\\
        &\le \frac{2}{ (M+2)^{\frac{3}{2} \left \lfloor \sqrt{\frac{M}{\log M}} \right \rfloor}} & \left (d = \left \lceil (M+2)^{3/2} \right \rceil \right )\\
        &\le \frac{2}{ \left ( \left \lfloor N^{2/5} \right \rfloor - 1 \right )^{\frac{3}{2}\left \lfloor \sqrt{\frac{\left \lfloor N^{2/5} \right \rfloor - 3}{\log (N^{2/5})}} \right \rfloor}} \,, & \left (M =  \left \lfloor N^{2/5} \right \rfloor - 3 \right )
    \end{align*}
    which concludes the proof.
\end{proof}

\section{Improved Analysis of Strategy Iteration Algorithm for \DSCV}
\label{section: si}
In this section, we first define the basic notions related to the strategy iteration algorithm (\Cref{section: si definitions}), then present the procedure \SI\ (\Cref{section: si algorithm}), and finally, we analyze the time complexity of the algorithm (\Cref{section: si analysis}).


\subsection{Basic Notions}\label{section: si definitions}

\noindent{\bf Best-response to a strategy.}
Given a discounted-payoff game $G$ with discount factor $\discfac$ and player-1 strategy $\strategyone$, the \emph{best-response} to $\strategyone$ for player~2 is a strategy $\strategytwo$ such that for all strategies $\strategytwo' \in \Strategytwo^P$ and all vertices $\vertex$, we have 
\[
    \discounted_\discfac \left (G_\vertex^{\strategyone, \strategytwo} \right ) \le \discounted_\discfac \left (G_\vertex^{\strategyone, \strategytwo'} \right )\,.
\]
The existence of the best-response strategies follows from~\Cref{Result: determinacy of tbgs}. 

\smallskip\noindent{\bf Bellman strategy extractor.} 
Given a discounted-payoff game $G$ with discount factor $\discfac$ and a function $f \colon \Vertices \to \RR$, Bellman strategy extractor is defined as follows:
\[
    \B_{G, \discfac}(f)(\vertex) = 
    \begin{cases}
    \arg \max_{\otherver \in \Edges(\vertex)} \reward(\vertex, \otherver) + \lambda f(\otherver),& \text{if } \vertex \in \Vertices_1\\
    
    \arg \min_{\otherver \in \Edges(\vertex)} \reward(\vertex, \otherver) + \lambda f(\otherver),& \text{if } \vertex \in \Vertices_2
    \end{cases}
\]
We assume that the ties are resolved independently of discount factor, e.g., given a fixed indexing of vertices, choosing the vertex with the least index.

\smallskip\noindent{\bf Polynomials for a strategy profile.} Given a discounted-payoff game, a vertex $\vertex$, and a strategy profile $(\strategyone, \strategytwo)$, we define a pair of polynomials $(P, Q)$ with integer coefficients such that $\frac{P(\discfac)}{Q(\discfac)}$ is the discounted value of play $G_\vertex^{\strategyone, \strategytwo}$ when the discount factor is $\discfac$. Given strategies $\strategyone$ and $\strategytwo$, the lasso-shaped play $G_\vertex^{\strategyone, \strategytwo}$ consists in a simple path $\calP \defas \langle \vertex_0, \ldots, \vertex_{p-1} \rangle$ and a cycle $\C \defas \langle \vertex_p, \ldots, \vertex_{p + c - 1} \rangle$ repeated forever. Then, 
\begin{align*}
    \discounted_\discfac(G_\vertex^{\strategyone, \strategytwo}) &= \sum_{i=0}^{\infty} \discfac^i \reward(\vertex_i, \vertex_{i+1})\\
    &= \sum_{i=0}^{p-1} \discfac^i \reward(\vertex_i, \vertex_{i+1}) + \sum_{i=p}^{\infty} \discfac^{i} \reward (\vertex_{i}, \vertex_{i+1})\\
    &= \sum_{i=0}^{p-1} \discfac^i \reward(\vertex_i, \vertex_{i+1}) + \frac{\sum_{i=0}^{c-1} \discfac^{p+i} \reward (\vertex_{p + i}, \vertex_{p + i+1})}{1 - \discfac^{c}}\\
    &= \frac{(1 - \discfac^c) \sum_{i=0}^{p-1} \discfac^i \reward(\vertex_i, \vertex_{i+1}) + \sum_{i=0}^{c-1} \discfac^{p+i} \reward (\vertex_{p + i}, \vertex_{p + i+1})}{1 - \discfac^c} \\
    &\sadef \frac{P(\discfac)}{Q(\discfac)} \,,
\end{align*}
where in the first equality we use the definition of $\discounted_\discfac$, in the second equality we partition the play into the simple path and the cycle, in the the third equality we use geometric sum, and in the the fourth we use algebraic manipulation.
Notice that the coefficients of $P$ and $Q$ are integers and bounded by $3W$. The degrees of $P$ and $Q$ are at most $n$.

\begin{example}
     \Cref{Figure: Turn-based game example} illustrates a turn-based game $G$ with three vertices: two player-1 vertices $a, b$ and one player-2 vertex $c$. The directed edges among them represent possible actions with associated weights: $1$ from $a$ to $b$, $3$ from $b$ to $c$, and $-2$ from $c$ to $b$. The initial vertex is $a$. Since each vertex has one possible action, there exist a positional strategy $\strategyone$ for player~1 and a positional strategy $\strategytwo$ for player~2. Given the strategy profile $(\strategyone, \strategytwo)$, the lasso-shaped play $G_a^{\strategyone, \strategytwo} = \langle a, b, c, b, c, \cdots \rangle$ consists in the simple $\calP = \langle a \rangle$ and the cycle $\C = \langle b, c \rangle$ repeated forever. For all $\discfac$, we have
    \[
        \discounted_\discfac(G_a^{\strategyone, \strategytwo}) = 1 + \frac{3\discfac - 2\discfac^2}{1 - \discfac^2} = \frac{1 + 3\discfac - 3\discfac^2}{1 - \discfac^2} = \frac{P(\discfac)}{Q(\discfac)}\,.
    \]
\end{example}

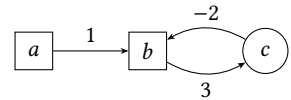
\begin{figure}[h]
\begin{tikzpicture}[
    node distance = 1cm,
    max/.style={draw, rectangle, minimum size=0.5cm, text centered},
    min/.style={draw, circle, minimum size=0.6cm, text centered},
]
\centering

\node (stmax) [max] {$a$};
\node (max) [max, right=of stmax] {$b$};
\node (min) [min, right=of max] {$c$};

\path (stmax) edge[->] node[above] {$1$} (max); 
\path (max) edge[->, bend right] node[below] {$3$} (min); 
\path (min) edge[->, bend right] node[above] {$-2$} (max); 

\end{tikzpicture}
\caption{A turn-based game.}
\label{Figure: Turn-based game example}
\end{figure}

\subsection{Strategy Iteration Algorithm}
\label{section: si algorithm}
Strategy iteration for turn-based games computes the optimal strategy for player~1. 
It starts with an arbitrary strategy $\strategyone^0$. Then, at each iteration, it locally improves player-1 strategy. 
Starting with strategy $\strategyone^k$ at iteration $k$, it computes the best-response to $\strategyone^k$ for player~2, called $\strategytwo^k$. 
This computation can be done in time $\calO(mn^2 \log m)$ by linear programming \cite{andersson2006improved}. 
Then, it improves the strategy $\strategyone^k$ by greedy local improvements using Bellman strategy extraction $\B_{G, \discfac}$. 
This procedure is guaranteed to reach a fixed point, which is an optimal strategy for player~1. 
The formal description is shown in \Cref{algorithm: strategy iteration}.
\begin{algorithm}[ht]
  \caption{Strategy Iteration}
  \label{algorithm: strategy iteration}
  \begin{algorithmic}[1]
    \Require Game $G$, discount factor $\discfac$, strategy $\strategyone^0$ for player 1
    \Ensure Optimal strategy $\strategyone^*$ for player 1
 
    \Procedure{\SI}{$G, \discfac, \strategyone^0$}
      \State $k \gets 0$
      \Repeat
        \State $\strategytwo^k \gets \text{the best response to } \strategyone^k$
        \State $\strategyone^{k+1} \gets \B_{G, \discfac}\left (\discounted_\discfac \left (G^{\strategyone^k, \strategytwo^k} \right ) \right )$
        \State $k \gets k + 1$
      \Until{$\strategyone^k = \strategyone^{k-1}$}
      \State \Return $\strategyone^k$
    \EndProcedure
  \end{algorithmic}
\end{algorithm}

\subsection{Improved Time Complexity Analysis}
\label{section: si analysis}
In this section, we analyze the time complexity of \SI\ algorithm (\Cref{Result: time complexity of SI for TBG-WUN}). 
We show that, given a turn-based game $G$, there exists a discount factor $\discfac_0$ such that, for all discount factors $\discfac \in [\discfac_0, 1)$, if we start \SI\ from the same initial strategy for both $\discfac_0$ and $\discfac$, \SI\ generates the same sequence of strategies (\Cref{Result: equivalence of sequences of strategies}). 
This result, alongside the time complexity of \SI\ in general case (\Cref{Result: time complexity of SI for TBG-BIN}), yields an improved complexity of \SI.
To prove \Cref{Result: equivalence of sequences of strategies}, we show that if a strategy profile outperforms another when the discount factor is $\discfac_0$, then it also outperforms when the discount factor is $\discfac \in [\discfac_0, 1)$ (\Cref{Result: ordering of strategy profiles}). \Cref{Example: Two-cycle game} illustrates that the performance of two strategy profiles can be compared by their associated rationals.

\begin{example}
\label{Example: Two-cycle game}
     \Cref{Figure: Two-cycle game example} illustrates a turn-based game $G$ with nine vertices, where all vertices are player-1 vertices. The directed edges among them represent possible actions with associated weights. The initial vertex is $a$. There are two positional strategies for player~1: $\strategyone_1$ (going left from $a$) and $\strategyone_2$ (going right from $a$). Since player~2 does not have any vertices, there is one positional strategy $\strategytwo$. For all discount factors $\discfac$, we have
     \[
         \discounted_\discfac \left (G_a^{\strategyone_1, \strategytwo} \right ) = \frac{\discfac + 2\discfac^4}{1 - \discfac^4}\,, \quad
         \discounted_\discfac \left (G_a^{\strategyone_2, \strategytwo} \right ) = \frac{2\discfac^2 + \discfac^3}{1 - \discfac^4}\,.
     \]
     The difference between the discounted payoff of $\strategyone_1$ and $\strategyone_2$ is
     \begin{align*}
         \discounted_\discfac \left (G_a^{\strategyone_1, \strategytwo} \right ) - \discounted_\discfac \left (G_a^{\strategyone_2, \strategytwo} \right ) &= \frac{\discfac - 2\discfac^2 - \discfac^3 + 2\discfac^4}{1 - \discfac^4}\\
         &= \frac{\discfac(1 + \discfac)(1 - \discfac)(1 - 2\discfac)}{1 - \discfac^4}\,.
     \end{align*}
     As $\discfac$ varies in $(0, 1)$, the performance of two strategies is compared as follows.
     \begin{itemize}
         \item If $\discfac \in (0, 1/2)$, then $\strategyone_1$ outperforms $\strategyone_2$.
         \item If $\discfac = 1/2$, then the performance of $\strategyone_1$ and $\strategyone_2$ are the same.
         \item If $\discfac \in (1/2, 1)$, then $\strategyone_2$ outperforms $\strategyone_1$.
     \end{itemize}
     
\end{example}

\begin{figure}[h]
\begin{tikzpicture}[
    node distance = 1cm,
    max/.style={draw, rectangle, minimum height=0.5cm, minimum width=0.5cm, text centered},
]

\node (center) [max] {$a$};
\node (left_1) [max, left=of center] {$b$};
\node (left_2) [max, above left=of left_1] {$c$};
\node (left_3) [max, below left=of left_2] {$d$};
\node (left_4) [max, below right=of left_3] {$d$};
\node (right_1) [max, right= of center] {$e$};
\node (right_2) [max, above right= of right_1] {$f$};
\node (right_3) [max, below right= of right_2] {$g$};
\node (right_4) [max, below left= of right_3] {$h$};

\path (center) edge[->,] node[above] {$0$} (left_1);
\path (center) edge[->,] node[above] {$0$} (right_1);

\path (left_1) edge[->,] node[above] {$1$} (left_2);
\path (left_2) edge[->,] node[above] {$0$} (left_3);
\path (left_3) edge[->,] node[below] {$0$} (left_4);
\path (left_4) edge[->,] node[below] {$2$} (left_1);

\path (right_1) edge[->,] node[above] {$0$} (right_2);
\path (right_2) edge[->,] node[above] {$2$} (right_3);
\path (right_3) edge[->,] node[below] {$1$} (right_4);
\path (right_4) edge[->,] node[below] {$0$} (right_1);
\end{tikzpicture}
\caption{A turn-based game with two cycles.}
\label{Figure: Two-cycle game example}
\end{figure}
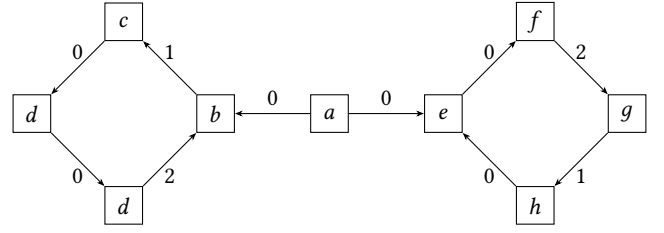

\begin{lemma}
\label{Result: ordering of strategy profiles}
    Consider a discounted-payoff game $G$, a vertex $\vertex$, and two strategy profiles $(\strategyone_1, \strategytwo_1)$ and $(\strategyone_2, \strategytwo_2)$. Fix discount factor 
    \[
        \discfac_0 \defas 1 - \left (24W(2n+1)^{7 W^{1/4} \sqrt{n} + 6} \right )^{-1}\,.
    \]
    If $\discounted_{\discfac_0}\left (G_\vertex^{\strategyone_1, \strategytwo_1} \right ) \ge \discounted_{\discfac_0}\left (G_\vertex^{\strategyone_2, \strategytwo_2} \right )$, then for all $\discfac \in [\discfac_0, 1)$, we have 
    \[
        \discounted_{\discfac}\left (G_\vertex^{\strategyone_1, \strategytwo_1} \right ) \ge \discounted_{\discfac}\left (G_\vertex^{\strategyone_2, \strategytwo_2}\right )\,.
    \]
\end{lemma}

\begin{proof}
    Let $(P_1, Q_1)$ (resp. $(P_2, Q_2)$) be the polynomials corresponding to a strategy profile $(\strategyone_1, \strategytwo_1)$ (resp. $(\strategyone_2, \strategytwo_2)$), and a vertex $\vertex$. We know that $\discounted_{\discfac_0}\left (G_\vertex^{\strategyone_1, \strategytwo_1} \right ) \ge \discounted_{\discfac_0}\left (G_\vertex^{\strategyone_2, \strategytwo_2} \right )$. Therefore, we have
    \[
        \frac{P_1(\discfac_0)}{Q_1(\discfac_0)} \ge \frac{P_2(\discfac_0)}{Q_2(\discfac_0)}\,,
    \]
    or equivalently,
    \[
      (P_1Q_2 - P_2Q_1)(\discfac_0) \ge 0 \,.
    \]
    We define 
    \[
        F \defas P_1Q_2 - P_2Q_1\,.
    \]
    Note $P_1$, $Q_1$, $P_2$, and $Q_2$ are polynomials of degree $n$ with integer coefficients bounded by $3W$, and $Q_1$ and $Q_2$ are of the form $1 - \lambda^{c_1}$ and $1 - \lambda^{c_2}$ where $c_1$ and $c_2$ are the length of the cycles in the lasso-shaped plays given the strategy profiles. Therefore, polynomial $F$ is of degree at most $2n$, and its coefficients are integer and are bounded by $12W$. If $F$ is identically 0, then for all $\discfac \in [0, 1)$, we have $\discounted_\discfac(G_\vertex^{\strategyone_1, \strategytwo_1}) = \discounted(G_\vertex^{\strategyone_2, \strategytwo_2})$. Otherwise,
    by \Cref{Result: bound on the roots of polynomials}-(\ref{Result: upper bound on the roots of polynomials}), for all roots $\lambda < 1$ of $F$, we have 
    \begin{align*}
        1 - \discfac &\ge \frac{\left \lfloor \frac{16}{7} (12W)^{1/4} \sqrt{2n} \right \rfloor !}{24W(2n+1)^{\frac{16}{7} (12W)^{1/4} \sqrt{2n} + 6}} & (\text{\Cref{Result: bound on the roots of polynomials}-(\ref{Result: upper bound on the roots of polynomials})}) \\
        &> \frac{1}{24W(2n+1)^{\frac{16}{7} (12W)^{1/4} \sqrt{2n} + 6}} \\
        &> \frac{1}{24W(2n+1)^{7 W^{1/4} \sqrt{n} + 6}}\,, & \left (\frac{16}{7}12^{1/4}\sqrt{2} < 7\right )
    \end{align*}
    which implies that $F$ does not have any roots in the interval $[\discfac_0, 1)$. Therefore, by continuity of $F$, for all discount factors $\discfac \in [\discfac_0, 1)$, we have 
    \[
        (P_1Q_2 - P_2Q_1)(\discfac) = F(\discfac) \ge 0\,,
    \]
    which implies that 
    \[
        \left (\frac{P_1}{Q_1} \right )(\discfac) \ge \left ( \frac{P_2}{Q_2} \right )(\discfac)\,.
    \]
    Hence, we have 
    \[
        \discounted_{\discfac}\left (G_\vertex^{\strategyone_1, \strategytwo_1} \right ) \ge \discounted_{\discfac}\left (G_\vertex^{\strategyone_2, \strategytwo_2}\right )\,,
    \]
    which concludes the proof.
\end{proof}

\begin{remark}
    If there is equality, by symmetry, one concludes that the functions of the respective profiles are equal in the interval~$[\discfac_0, 1)$.
\end{remark}

The above result shows that for any vertex $\vertex$, the ordering of strategy profiles is the same for both discount factors $\discfac_0$ and $\discfac$. \Cref{Result: ordering of strategy profiles} yields that, if \SI\ starts from the same initial strategy for both $\discfac_0$ and $\discfac$, then \SI\ outputs the same optimal strategies. 

 \begin{corollary}
 \label{Result: equivalence of sequences of strategies}
    Fix discount factor 
    \[
        \discfac_0 \defas 1 - \left (24W(2n+1)^{7 W^{1/4} \sqrt{n} + 6} \right )^{-1}\,.
    \]
    Consider a discounted-payoff game $G$ with discount factor $\discfac \in [\discfac_0, 1)$, and player-1 strategy $\strategyone^0$. 
    Let $\strategyone^k$ and $\overline{\strategyone}^k$ be the strategy for player~1 after $k$ iterations of $\, \Call{\SI}{G, \discfac_0, \strategyone^0}$ and $\, \Call{\SI}{G, \discfac, \strategyone^0}$, respectively. Then, for all $k \ge 0$, we have $\strategyone^k = \overline{\strategyone}^k$. In particular,  
    \[
        \Call{\SI}{G, \discfac_0, \strategyone^0} = \Call{\SI}{G, \discfac, \strategyone^0}\,.
    \]
 \end{corollary}

 \begin{proof}
  We proceed with the proof by induction on the number of iterations. 
  
  \smallskip\noindent{\em Base case $k = 0$.} Since the initial strategy for both procedure calls is the same, then we have $\strategyone^0 = \overline{\strategyone}^0$, which completes the case.

  \smallskip\noindent{\em Induction case $k > 0$.}
  We claim that $\strategytwo^k = \overline{\strategytwo}^k$, where $\strategytwo^k$ and $\overline{\strategytwo}^k$ are the best responses to $\strategyone^k$ and $\overline{\strategyone}^k$ for player~2 with respect to $\discfac$ and $\discfac_0$, respectively. 
  Indeed, since $\strategytwo^k$ is the best response to $\strategyone^k$ for player~2, then, for all vertices $v$, we have 
  \[
    \discounted_{\discfac_0}\left (G_\vertex^{\strategyone^k, \strategytwo^k} \right ) \le \discounted_{\discfac_0}\left (G_\vertex^{\strategyone^k, \overline{\strategytwo}^k} \right )\,.
  \]
  by \Cref{Result: ordering of strategy profiles}, we have 
  \[
    \discounted_{\discfac}\left (G_\vertex^{\strategyone^k, \strategytwo^k} \right ) \le \discounted_{\discfac} \left (G_\vertex^{\strategyone^k, \overline{\strategytwo}^k} \right)\,
  \]
  The inequality also holds in the other direction since $\overline{\strategytwo}^k$ is the best response to $\overline{\strategyone}^k$. Since ties are broken independently of the discount factor, we have that $\strategytwo^k = \overline{\strategytwo}^k$. Similarly, we can show that $\strategyone^{k+1} = \overline{\strategyone}^{k+1}$, which concludes the induction case and yields the result.
 \end{proof}

\begin{remark}
    
\end{remark}

 \begin{proof}[Proof of \Cref{Result: time complexity of SI for TBG-WUN}]     
   Fix 
   \[
    \lambda_0 \defas 1 - \left (24W(2n+1)^{7 W^{1/4} \sqrt{n} + 6} \right )^{-1}\,.
   \]
   We proceed with a proof by cases on the size of $\lambda$. 
 \begin{itemize}
    \item Case $\lambda \le \lambda_0$. By \Cref{Result: time complexity of SI for TBG-BIN}, the procedure \Call{\SI}{} terminates after $n^{\calO \left ( W^{1/4} \sqrt{n}  \right )}$ iterations, which completes the case.
    \vspace{0.1em}
    \item Case $\lambda > \lambda_0$. By \Cref{Result: equivalence of sequences of strategies}, the procedure \Call{\SI}{} for discount factors $\lambda$ and $\lambda_0$ terminates after the same number of iterations. Therefore, it terminates after $n^{\calO \left ( W^{1/4} \sqrt{n}  \right )}$ iterations, which concludes the case and yields the result.
 \end{itemize}
 \end{proof}

\section{Discussion and conclusion} 
We discuss the novelty of our work and future directions.
The novelty of our work is not a new algorithm, but a new improved analysis of 
a classical and simple algorithm (\SI\ algorithm) that has been widely studied.
Moreover, most results in the literature focus on establishing lower bounds 
for the \SI\ algorithm~\cite{friedman2009super,friedmann2011subexponential,fearnley2015complexity}; in contrast we focus on a better upper bound. 
The novel aspects of our analysis are as follows: 
(a)~establishing the connection of analysis of \SI\ with a problem about lower bounds 
on roots of a class of polynomials, which is the key insight; and
(b)~establishing lower and upper bounds on roots of the required class of polynomials.
While our analysis only requires the lower bounds on roots, the significance of the upper 
bounds on roots is to show that our technique does not yield a polynomial-time bound without 
further non-trivial insights.
Our result shows that for the \TBGW\ problem we have a deterministic sub-exponential time
algorithm. Discounted-sum games lie in between mean-payoff games and stochastic games with reachability objectives,
i.e., linear-time reduction exists from mean-payoff to discounted-sum games and 
from discounted-sum games to stochastic games, but no reductions are known for the converse
direction. 
Whether discounted-sum games are more similar to mean-payoff games or stochastic games
is an intriguing question. 
When the weights are expressed in unary, mean-payoff games admit polynomial-time algorithm,
whereas stochastic games with rewards~0 and~1, and all probabilities are half are as hard
as general stochastic games.
Thus deterministic sub-exponential time bound for unary stochastic games is a major open 
question.
Our result shows the difference of unary discounted-sum games compared to unary stochastic games
by establishing deterministic sub-exponential time upper bound.
Whether there is a polynomial-time algorithm for unary discounted-sum games or there is a deterministic
sub-exponential time algorithm for unary stochastic games are interesting questions for future work.

\smallskip\noindent{\bf Acknowledgements.} This research was partially supported by the ERC CoG 863818 (ForM-SMArt) grant.

\bibliographystyle{acm}
\bibliography{refs}


\end{document}